\documentclass[envcountsame]{article}

\usepackage{url}
\usepackage{amsmath,amsthm} 
\usepackage{amssymb,amsmath}
\usepackage{tikz}
\usetikzlibrary{fit}
\usepackage[english]{babel}
\usepackage[colorlinks=true,citecolor=black,linkcolor=black]{hyperref}
\usepackage{a4wide}
\usepackage[color=green!40]{todonotes}
\usepackage{xspace} 
\usepackage{fancybox}			
\usepackage{mathtools}	
\usepackage{nicefrac}
\usepackage{booktabs}

\newtheorem{theorem}{Theorem} 
\newtheorem{proposition}[theorem]{Proposition}%
\newtheorem{corollary}[theorem]{Corollary}
\newtheorem{definition}[theorem]{Definition}%

\newcommand{\np}{\mathsf{NP}}
\newcommand{\wone}{\mathsf{W[1]}}
\newcommand{\wtwo}{\mathsf{W[2]}}
\newcommand{\fpt}{\mathsf{FPT}}
\newcommand{\xp}{\mathsf{XP}}

\newcommand{\nat}{\mathbb{N}}
\newcommand{\p}{\mathsf{P}}

\newcommand{\R}{\mathcal{R}}

\renewcommand\leq\leqslant
\renewcommand\geq\geqslant

\def\mcis{\textsc{Maximum Common Induced Subgraph}\xspace}
\def\mccis{\textsc{Maximum Common Connected Induced Subgraph}\xspace}
\def\isi{\textsc{Induced Subgraph Isomorphism}\xspace}
\def\icsi{\textsc{Induced Connected Subgraph Isomorphism}\xspace}

\def\ISIC{\textsc{I(C)SI}\xspace}

\def\clique{\textsc{Clique}\xspace}

\tikzstyle{bigvertex}=[circle,fill=black!0,minimum size=15pt,inner sep=0pt]
\tikzstyle{rect}=[rectangle, rounded corners]
\tikzstyle{edge} = [draw,-,rounded corners=8pt]
\tikzstyle{vertex}=[circle, draw, inner sep=0pt, minimum width=15pt]
\usetikzlibrary{trees,decorations,decorations.pathmorphing,decorations.pathreplacing}
\usetikzlibrary{shapes}

\newcommand{\PbOpt}[3]{%
\begin{center}
  \begin{tabular}{|l|}%
  \hline
    \begin{minipage}[c]{.95\textwidth}
    \smallskip%
      \par\noindent%
      \textsc{#1:}
      \par\noindent%
      $\bullet$
      \textbf{\textsf{Input}}: #2%
      \par\noindent%
      $\bullet$
      \textbf{\textsf{Output}}: #3%
      \smallskip%
      \par\noindent%
    \end{minipage}
    \\\hline
  \end{tabular}%
\end{center}
}%

\begin{document}

\title{On the Complexity of Various Parameterizations of Common Induced Subgraph Isomorphism\footnote{An extended abstract of this work appears in~\cite{DBLP:conf/iwoca/Abu-KhzamBS14}}
}


\author{Faisal N. Abu-Khzam\footnote{\noindent Lebanese American University, Beirut, Lebanon \tt{faisal.abukhzam@lau.edu.lb}} \and \'Edouard Bonnet\footnote{ Middlesex University, Department of Computer Science, Hendon, London, UK \tt{edouard.bonnet@lamsade.dauphine.fr}} \and Florian Sikora\footnote{Universit\'{e} Paris-Dauphine, PSL Research University, CNRS, LAMSADE, Paris, France,  \tt{florian.sikora@dauphine.fr}}
}

\date{}

\maketitle
\begin{abstract}
In the \mcis problem (henceforth MCIS), given two graphs $G_1$ and $G_2$, one looks for a graph with the maximum number of vertices being both an induced subgraph of $G_1$ and $G_2$.
MCIS is among the most studied classical $\np$-hard problems. 
It remains $\np$-hard on many graph classes including forests. 
In this paper, we study the parameterized complexity of MCIS.
As a generalization of \textsc{Clique}, it is $\wone$-hard parameterized by the size of the solution.
Being NP-hard even on forests, most structural parameterizations are intractable.
One has to go as far as parameterizing by the size of the minimum vertex cover to get some tractability.
Indeed, when parameterized by $k := \text{vc}(G_1)+\text{vc}(G_2)$ the sum of the vertex cover number of the two input graphs, the problem was shown to be fixed-parameter tractable, with an algorithm running in time $2^{O(k \log k)}$.
We complement this result by showing that, unless the ETH fails, it cannot be solved in time $2^{o(k \log k)}$.
This kind of tight lower bound has been shown for a few problems and parameters but, to the best of our knowledge, not for the vertex cover number.
We also show that MCIS does not have a polynomial kernel when parameterized by $k$, unless $\np \subseteq \mathsf{coNP}/poly$. 
Finally, we study MCIS and its connected variant MCCIS on some special graph classes and with respect to other structural parameters.
%

\end{abstract}


\section{Introduction}
\label{intro}

A common induced subgraph of two graphs $G_1$ and $G_2$ is a graph that is isomorphic to an induced subgraph of both graphs.
The problem of finding a common induced subgraph with the maximum number of vertices (or edges) has many applications in a number of domains including bioinformatics and chemistry \cite{Grindley1993,Koch1996,MW81,RW02,Yamaguchi2004}. 
In the decision version of the problem, we are given an integer $k$ and the question is to decide whether there is a solution with at least $k$ vertices.
We say that the solution size $k$ is the \emph{natural parameter} of the 
problem. 

Concerning its classical complexity, \mcis is $\np$-complete, and remains so 
on forests. When the common subgraph is required to be connected, the problem 
is in $\p$ for trees \cite{Garey1979}.
Moreover, \mcis is also in $\p$ when the two input graphs are connected and 
(both) have bounded treewidth and bounded degree \cite{Akutsu1993}. 

A particular case of \mcis is the well known \isi (ISI) decision problem, where the question posed is whether $G_1$ is isomorphic to an induced 
subgraph of $G_2$.
In other words, it is equivalent to \mcis where $k= |G_1|$.
In this case, $G_1$ is called the pattern graph and $G_2$ is called the host 
graph.
ISI is known to be $\np$-hard, even when $G_2$ is an interval graph and $G_1$ 
is a proper interval graph, but it becomes polynomial-time solvable when $G_1$ 
is in addition connected~\cite{Heggernes}. 
Unlike \textsc{Subgraph Isomorphism}, \isi remains $\np$-hard when both
the host graph and the pattern graph consist of a disjoint union of 
paths \cite{Damaschke1991}. 
From the parameterized complexity viewpoint, the problem is $\wone$-hard in 
general for the natural parameter, by a straightforward reduction 
from \textsc{$k$-Clique}. Therefore MCIS is also $\wone$-hard. 
Moreover, ISI (and, therefore, MCIS) remains $\wone$-hard even when both 
graphs are interval graphs~\cite{MarxS13}. 
On the other hand, ISI is $\fpt$ on nowhere dense 
graphs, being expressible by a first-order formula of length
function of the natural parameter $k$~\cite{GroheKS14}.
This generalizes what was previously known about ISI on 
$H$-minor free graphs~\cite{FG01} 
and graphs of bounded degree~\cite{Cai2006}. 
We observe that whenever ISI in $\fpt$ on a certain graph class, then so 
is MCIS. 
To see this, note that an arbitrary instance $(G_1,G_2,k)$ of MCIS can be 
reduced in fpt-time to instances of ISI by enumerating each graph $H$ on 
$k$ vertices and checking whether $H$ is an induced subgraph of $G_1$ and $G_2$.
This implies that ISI and MCIS have the same parameterized complexity with 
respect to the natural parameter.  

Another way of dealing with the hardness of a problem is to study its complexity with respect to auxiliary (or structural) parameters, to better  understand its algorithmic behavior (see for example~\cite{FellowsJR13}).
Being $\np$-hard on disjoint union of chordless paths~\cite{Damaschke1991}, 
MCIS is hard on graphs with bounded treewidth as well as graphs where the size of the minimum feedback vertex set is bounded. 
Thus the problem is paraNP-hard when parameterized by the treewidth of the input graphs, or by a bound on the sizes of their minimum feedback vertex sets.
Therefore, we need to look for ``bigger'' parameters.
And indeed, if the parameter $k$ is a bound on the sizes of the minimum vertex 
covers of the input graphs, 
then the problem is in $\fpt$, with a running time of $O((24k)^k)=2^{O(k \log k)}$~\cite{AbuKhzam2014}.
In this paper, we show that this algorithm cannot be significantly improved: 
unless the Exponential Time Hypothesis 
(ETH) fails, there is no algorithm solving MCIS in time $O^*(2^{o(k \log k)})$, where the $O^*$ notation suppresses the polynomial factors.
We 
also prove 
that MCIS does not have a polynomial-size kernel in this case unless $\np \subseteq \mathsf{coNP}/poly$. 
These two latter results answer open problems raised in \cite{AbuKhzam2014}.
Finally, we show that \mccis (MCCIS), where the solution should be a connected 
graph, is also fixed-parameter tractable when parameterized by $k := vc(G_1)+vc(G_2)$.

\section{Preliminaries}
\label{prelims}

Two finite graphs $G_1=(V_1,E_1)$ and $G_2=(V_2,E_2)$ are \textit{isomorphic} 
if there 
is a bijection $\pi: V_1 \to V_2$ such that $\forall u,v \in V_1: 
uv \in E_1 \Leftrightarrow \pi(u)\pi(v) \in E_2$. Given a graph $G=(V,E)$, a graph  
$G'=(V',E')$ is an \textit{induced subgraph} of $G$ if $V' \subseteq V$ and 
$E' = E(V')=\{uv \in E $ $|$ $ u, v \in V' \}$, i.e. $E'$ is the edge set with both 
extremities in $V'$. We also say that $G'$ is the subgraph of $G$ induced by $V'$. 

The \textit{girth} of a graph $G$ is the length of the shortest cycle contained in $G$. 
Contracting an edge $uv$ consists of deleting $uv$ and replacing the vertices $u$ and $v$ by a single vertex $w$ in the incidence relation (edges incident on $u$ or $v$ become incident on $w$).
A graph $H$ is a \textit{minor} of graph $G$ if $H$ is obtained from a subgraph of $G$ by applying zero or more edge contractions. Given a fixed graph $H$, a family $\cal{F}$ of graphs is said to be \textit{$H$-minor free} if $H$ is not a minor of any element of $\cal{F}$.

The \mcis problem is defined formally as follows. 

\PbOpt{\mcis (MCIS)}{Two graphs $G_1=(V_1,E_1)$ and $G_2=(V_2,E_2)$.}{An induced 
subgraph $G_1'$ of $G_1$ isomorphic to an induced subgraph $G_2'$ of $G_2$ with a 
maximum number of vertices.}



\mccis (MCCIS) is defined as MCIS with the additional restriction that the solution must be \emph{connected}. 

For completeness, we also give the definition of \isi:

\PbOpt{\isi (ISI)}{Two graphs $G_1=(V_1,E_1)$ and $G_2=(V_2,E_2)$.}{An induced 
subgraph $G'_2$ of $G_2$ isomorphic to $G_1$ if it exists.}

\icsi (ICSI) is defined as ISI but $G_1$ must be connected. 

\paragraph{Parameterized complexity}
A parameterized problem $(I,k)$ is \textit{fixed-parameter tractable} (or in 
the class $\fpt$) with respect to parameter $k$ if it can be solved in 
$f(k)\cdot|I|^c$ time (i.e. in fpt-time), where 
$f$ is any computable function and $c$ is a constant (see 
\cite{Downey2013,Niedermeier2006} 
for more details about fixed-parameter tractability).
The parameterized complexity hierarchy is composed of the classes 
$\fpt \subseteq \wone 
\subseteq \wtwo \subseteq \dots \subseteq \mathsf{XP}$. 
The class $\xp$ contains problems solvable in time $f(k)\cdot |I|^{g(k)}$, 
where $f$ and $g$ are unrestricted functions.
A problem is said to be \emph{paraNP-hard} if it is $\np$-hard even for a 
constant value of the parameter (it hence cannot be in $\mathsf{XP}$).
A $\wone$-hard problem is not fixed-parameter tractable, unless $\fpt=\wone$, and one can prove $\wone$-hardness by means of a \emph{parameterized reduction} from a $\wone$-hard problem.
This is a mapping of an instance~$(I,k)$ of a problem~$A_1$ in $g(k)\cdot |I|^{O(1)}$ time (for any computable function~$g$) into an instance $(I',k')$ for~$A_2$ such that $(I,k)\in A_1\Leftrightarrow (I',k')\in A_2$ and $k'\le h(k)$ for some function~$h$.

A powerful technique to design parameterized algorithms is \textit{kernelization}.
In short, kernelization is a polynomial-time self-reduction algorithm that takes an instance $(I,k)$ of a parameterized problem $P$ as input and computes an equivalent instance $(I',k')$ of $P$ such that $|I'| \leqslant h(k)$ for some computable function $h$ and $k' \leqslant k$.
The instance $(I',k')$ is called a \textit{kernel} in this case.
If the function $h$ is polynomial, we say that $(I',k')$ is a polynomial kernel.
It is well known that a problem is in $\fpt$ iff it has a kernel, but this equivalence yields super-polynomial kernels (in general).
To design efficient parameterized algorithms, a kernel of polynomial (or even linear) size in $k$ is important.
However, some lower bounds on the size of the kernel can be shown unless some polynomial hierarchy collapses.
To show this result, we will use the cross composition technique developed by Bodlaender et al.~\cite{Bodlaender2014}.

\begin{definition}[Polynomial equivalence relation \cite{Bodlaender2014}]
  An equivalence relation $\R$ on $\Sigma^*$ is said to be \emph{polynomial} if the following two conditions hold:
  (i) There is an algorithm that given two strings $x, y \in \Sigma^*$ decides whether $x$ and $y$ belong to the same equivalence class in time $(|x| + |y|)^{O(1)}$.
  (ii) For any finite set $S \subseteq \Sigma^*$ the equivalence relation $\R$ partitions the elements of $S$ into at most $(\max_{x \in S} |x|)^{O(1)}$ classes.

\end{definition}

\begin{definition}[OR-cross-composition \cite{Bodlaender2014}] Let $L \subseteq \Sigma^*$ be 
a set and let $Q \subseteq \Sigma^* \times \mathbb{N}$ be a parameterized problem. We say 
that $L$ \emph{cross-composes} into $Q$ if there is a polynomial equivalence relation 
$\R$ and an algorithm which, given $t$ strings $x_1, x_2,  \dots, x_t$ belonging to the 
same equivalence class of $\R$, computes an instance $(x^*,k^*) \in \Sigma^* \times \mathbb{N}$ 
in time polynomial in $\sum_{i=1}^{t}|x_i|$ such that: (i) $(x^*,k^*) \in Q \Leftrightarrow x_i \in L$ for some $1 \leqslant  i \leqslant t$. (ii) $k^*$ is bounded by a polynomial in $\max_{i=1}^t |x_i| + \log t$.

\end{definition}

\begin{proposition}[\cite{Bodlaender2014}]
Let $L \subseteq \Sigma^*$ be a set which is $\np$-hard under Karp reductions. If 
$L$ cross-composes into the parameterized problem $Q$, then $Q$ has no polynomial 
kernel unless $\np \subseteq \mathsf{coNP}/poly$.
\end{proposition}

The \emph{Exponential Time Hypothesis} (ETH) is a conjecture by Impagliazzo et al. asserting that there is no $2^{o(n)}$-time algorithm for \textsc{3-SAT} on instances with $n$ variables \cite{ImpagliazzoETH}. 
The ETH, together with the sparsification lemma \cite{ImpagliazzoETH}, even implies that there is no $2^{o(n+m)}$-time algorithm solving \textsc{3-SAT}.
Many algorithmic lower bounds have been proved under the ETH, see for example~\cite{LokshtanovSODA}.

We say that a parameterized problem is {\em fixed-parameter enumerable} if all feasible solutions can be enumerated in $O(f(k)|I|^c)$ time, where $f$ is a computable function of the parameter $k$ only, and $c$ is a constant.

\section{Parameterized Complexity with respect to the natural parameter}\label{sec:nat}

We study the parameterized complexity of \isi, \mcis, \icsi, and \mccis with 
respect to the natural parameter. 
We will in particular study these problems in 
graphs of bounded degeneracy, chordal graphs, and graphs of large girth.


\begin{theorem}\label{thm:complete}
MCIS, MCCIS, ISI, and ICSI are $\wone$-complete.
\end{theorem}

\begin{proof}
Since those problems are $\wone$-hard by a straightforward reduction from \textsc{$k$-Clique}, it suffices to show membership in $\wone$.
In \cite{Cesati03}, it is shown that if a problem can be reduced in FPT time 
to simulating a non-deterministic single-taped Turing machine halting in at 
most $f(k)$ steps, for some function $f$, then it is in $\wone$.
The Turing machine can have an alphabet and a set of states of size depending 
on the size of the input of the initial problem.
In our case, we can design a Turing machine that guesses in $2k$ steps the 
corresponding right $k$ vertices in $G_1$ (for I(C)SI this part is not 
necessary) and the right $k$ vertices in $G_2$ (our alphabet being isomorphic 
to an indexing of $V(G_1) \cup V(G_2)$) and then check in time $O(k^2)$ whether 
the two induced subgraphs are isomorphic (and that they are connected for ICSI 
and MCCIS).
\end{proof}

%
%

In \cite{MS09} it was shown that 
\textsc{Maximum Induced Matching}\footnote{where one looks for a largest 
subset of vertices that induce a disjoint union of edges} is $\wone$-hard 
on bipartite graphs. This implies that MCIS is $\wone$-hard on bipartite 
graphs. In fact, we show that MCIS remains $\wone$-hard on more restricted 
graph classes, namely $C_4$-free bipartite graphs with degeneracy $2$.
In particular, those graphs have girth at least $6$.
This result tells us two things about MC(C)IS.
The first is that the fixed-parameter algorithm of Cai et al.~\cite[Theorem 1]{Cai2006} 
cannot be extended from bounded degree to bounded degeneracy (note that some W-hard problems on general graphs become $\fpt$ on graphs with bounded degeneracy, such as the W[2]-complete \textsc{Dominating Set} problem~\cite{AlonG09}).
The second is that short cycles are not making MC(C)IS W[1]-hard; they are W[1]-hard even without them.
In \cite{RamanS08}, the authors present fixed-parameter algorithms on graphs 
of girth 5, for some problems which are W-hard on general graphs.
MCIS and MCCIS are also resistant to this approach.

\begin{theorem}
\isi and \icsi are $\wone$-complete even when both graphs are $C_4$-free 
bipartite graphs with degeneracy at most $2$.
\end{theorem}
\begin{proof}
The incidence graph $I(G)$ of any graph $G=(V,E)$, obtained by subdividing each edge of $G$ once, has degeneracy $2$. 
Indeed, graph $I(G)$ is the bipartite graph $(V \uplus E, F)$ where the edges of $F$ are all the $ue$ for which $u \in V$, $e \in E$, and $u$ is an endpoint of $e$. 
All the vertices $e \in E$ of $I(G)$ have degree $2$.
Therefore, they can be removed first. 
Then, what is left in $I(G)$ is the independent set $V$.

We transform any input $G=(V,E), k > 3$ of \textsc{$k$-Clique}, into the instance $I(K_k),I(G)$ of \ISIC, where both graphs have degeneracy $2$.
The problem consists of finding the incidence graph of a $k$-clique within the incidence graph of $G$. 
We show that it is equivalent to finding a $k$-clique in $G$.
Obviously, if there is a $k$-clique $S$ in $G$, then the graph $I(G)[S \cup E(S)]$ is isomorphic to $I(K_k)$.
Now, let us assume that $I(K_k)$ is isomorphic to an induced subgraph of $I(G)$.
We denote by $a_1, \ldots, a_k$ the vertices of $I(K_k)$ with degree $k-1$, and by $b_1, \ldots, b_{{k \choose 2}}$ the vertices of $I(K_k)$ with degree $2$.
We denote by $\psi: V(I(K_k)) \rightarrow V(I(G))$ the isomorphism from graph $I(K_k)$ to an induced subgraph of $I(G)$.
Let $u_i=\psi(a_i)$ for each $i \in [k]$, and $v_j=\psi(b_j)$ for each $j \in [{k \choose 2}]$.
We set $S=\{u_1, \ldots, u_k, v_1, \ldots, v_{{k \choose 2}}\}$.
For every $i \in [k]$, $u_i \in V$ since the degree of $a_i$ in $I(K_k)$ is $k-1 > 2$ (hence, the degree of $u_i$ in $S$ is also $k-1>2$).
Now, for every $j \in [{k \choose 2}]$, $v_j \in E$ since $v_j$ has two neighbors in $V$ (recall that $I(G)$ is bipartite).
Therefore, $u_1, \ldots, u_k$ are $k$ vertices in $V$ inducing precisely ${k \choose 2}$ edges.
Hence, $\{u_1, \ldots, u_k\}$ is a $k$-clique in $G$.

Membership in $\wone$ comes from \autoref{thm:complete}. 
\end{proof}

\begin{corollary}
\mcis and \mccis remain $\wone$-complete on bipartite graphs of girth 6 and 
degeneracy 2.
\end{corollary}

The absence of triangles and cycles of length four in the input graphs does not make the problems tractable.
We show that the absence of a long induced cycle does not help either (in~\cite{Arumugam11}, the authors show that the $\wtwo$-hard problem \textsc{Dominator Coloring} is in $\fpt$ when the input graph is chordal).
More specifically, all four problems are $\wone$-hard on chordal graphs.
In fact, we can even show that these problems remain $\wone$-hard on a proper subclass of chordal graphs called split graphs.
A split graph is a graph whose vertex set can be partitioned into a set inducing a clique and an independent set.

\begin{theorem}
ISI (hence MCIS) and ICSI (hence MCCIS) remain $\wone$-hard on split graphs.
\end{theorem}

\begin{proof}
Similarly to the previous construction, we define $I'(G)$ as the graph 
$(V \uplus E, F)$ where the edges of $F$ are the edges $ue$ for which 
$u \in V$, $e \in E$, and $u$ is an endpoint of $e$, plus all the edges 
$uv$ with $u,v \in V$. 
The graph $I'(G)$ is split: $V$ induces a clique in $I'(G)$ and $E$ induces 
an independent set.
From an instance $G$ of \textsc{$k$-Clique} with $k > 3$, we build the 
equivalent instance $I'(K_k),I'(G)$ of MC(C)IS and I(C)SI.
The soundness can be obtained in the same way as in the previous proof.
\end{proof}

Let us now say some words about the complexity of the connected version.
First we note that MCIS is $\np$-hard on forests while MCCIS is solvable in 
polynomial-time in this case: given two forests $G_1$ and $G_2$, run the 
polynomial-time MCIS algorithm of Akutsu on every pair of trees 
from $G_1$ and $G_2$~\cite{Akutsu1992}.
From the parameterized complexity standpoint,
\mccis is $\fpt$ whenever \isi is $\fpt$ 
since the enumeration of all $O(2^{k^2})$ possible induced \emph{connected} subgraphs can be used as described in the introduction.
The converse is true on classes of graphs which are closed by adding a universal vertex (i.e., a vertex linked to all the other vertices).
An instance $(G_1,G_2,k)$ of ISI can be reduced to an equivalent instance $(G_1',G_2',k+1)$ of MCCIS by letting $G_i'$ be the graph obtained by adding a vertex to $G_i$ that is made adjacent to all other vertices of $G_i$.

\section{Structural parameterization}\label{sec:struct}

Let us first recall that $\text{tw}(G) \leqslant \text{fvs}(G) + 1 \leqslant \text{vc}(G) + 1$, where $\text{tw}(G)$ (resp. $\text{fvs}(G)$, $\text{vc}(G)$) represents the treewidth (resp. the feedback vertex set number, the vertex cover number) of $G$~\cite{Fellows2013}.
As noted before, if the parameter is the combination of $\text{tw}(G_1)$ and $\text{tw}(G_2)$ then MCIS is known to be $\wone$-hard.
Even more, if the parameter is the combination of $\text{fvs}(G_1)$ and $\text{fvs}(G_2)$ (which is bigger than the combination of the treewidth), then the problem is not even in $\xp$ since \mcis and \isi are $\np$-hard on disjoint union of chordless paths, a case where the parameter is equal to 0~\cite{Damaschke1991,Garey1979}.


\begin{theorem}[\cite{Damaschke1991,Garey1979}]
\label{notinXP}
\mcis is paraNP-hard when parameterized by $\text{fvs}(G_1)+\text{fvs}(G_2)$ (and hence by $\text{tw}(G_1)+\text{tw}(G_2)$).
\end{theorem}



One can 
extend this result to make it valid for the connected version.

\begin{theorem}\label{th:fvsone}
\icsi, and as a corollary \mccis, are 
 paraNP-hard when parameterized by $\text{fvs}(G_1)+\text{fvs}(G_2)$.
\end{theorem}

\begin{proof}
  Given an instance of \isi on forests $G_1$ and $G_2$ (each with at least 2 trees), we build an instance of \icsi by adding a universal vertex (connected to every node) in $G_1$ and in $G_2$.
  Both graph have thus a feedback vertex set of value one.
  One can see that these two universal vertices must be matched together since they are the only ones with sufficiently high degree.
  Then, there is a solution for \isi iff there is a solution for \icsi. 
The result of course holds for MCCIS, too.
\end{proof}





\medskip

It was shown in \cite{AbuKhzam2014} that MCIS is in $\fpt$ if the parameter is the combination of $\text{vc}(G_1)$ and $\text{vc}(G_2)$. Accordingly, the problem has a kernel, but no polynomial bound is known on its size. We show that, in this case,
the kernel cannot be polynomial unless $\np \subseteq \mathsf{coNP}/poly$.

\begin{theorem}\label{th:no poly kernel}
Unless $\np \subseteq \mathsf{coNP}/poly$, \mcis has no polynomial kernel when parameterized by the sum of the sizes of vertex covers in the two input graphs.
\end{theorem}

\begin{proof}

We will define an OR-cross-composition from the $\np$-complete \clique, problem, where the given instance is a tuple $(G^c,l)$ and the question is whether the graph $G^c$ contains a clique on $l$ vertices. 

Given $t$ instances, $(G^c_1,l_1), (G^c_2,l_2), \dots, (G^c_t,l_t)$, of \clique, where $G^c_i$ is a graph and $l_i \in \nat, \forall 1 \leqslant i \leqslant t$, we define our equivalence relation $\R$ such that any strings that are not encoding valid instances are equivalent, and $(G^c_i,l_i), (G^c_j,l_j)$ are equivalent iff $|V(G^c_i)| = |V(G^c_j)|$, and $l_i = l_j$. Hereafter, we assume that $V(G^c_i) = \{1,\dots,n\}$ and $l_i = l$, for any $1 \leqslant i \leqslant t$. We will build an instance of \mcis parameterized by the vertex cover $(G_1,G_2,l',Z)$ where $G_1$ and $G_2$ are two graphs, $l' \in \nat$ and $Z \subseteq V(G_2)$ is a vertex cover of $G_2$ computed in fpt-time, such that there is a solution of size $l'$ for \mcis iff there is an $i, 1 \leqslant i \leqslant t$ such that there is a solution of size $l$ in $G^c_i$. We will now describe how to build $G_1$ and $G_2$.\\

To build $G_2$ (see also Figure~\ref{fig:g2}):

\begin{itemize}
\item $V(G_2) = \{p,q,r\} \cup \{a_i $ $|$ $ 1 \leqslant i \leqslant t\} \cup \{e_{uv} $ $|$ $ 1 \leqslant u < v \leqslant n\} \cup \{x_i $ $|$ $1 \leqslant i \leqslant n\}$,
\item $E(G_2)_1 = \{pq,pr,qr\}$,
\item $E(G_2)_2 = \{ra_i $ $|$ $ 1 \leqslant i \leqslant t\}$,
\item $E(G_2)_3 = \{a_ie_{uv} $ $|$ $ uv \in E(G^c_i) \}$,
\item $E(G_2)_4 = \{e_{uv}x_u, e_{uv}x_v $ $|$ $ \forall 1 \leqslant u < v \leqslant n\}$,
\item $E(G_2) = E(G_2)_1 \cup E(G_2)_2 \cup E(G_2)_3 \cup E(G_2)_4$.
\end{itemize}

\begin{figure}[ht!]
\centering
\begin{tikzpicture}[scale=0.9,auto]

\node[draw] (r) at (-1,0) {$r$};
\node[draw] (p) at (-1.7,0.5) {$p$};
\node[draw] (q) at (-0.3,0.5) {$q$};

\draw (p) -- (r) -- (q);
\draw (p) -- (q);

\node[draw] (t1) at (-3,-1) {$a_1$};
\node[draw] (t2) at (-2,-1) {$a_2$};
\node () at (-0.5,-1) {$\dots$};
\node[draw] (tt) at (1,-1) {$a_t$};
\node[draw,rectangle,rounded corners, fit = (t1) (tt)]  () {};

\node[draw] (e12) at (-4,-3) {$e_{1,2}$};
\node[draw] (e13) at (-3,-3) {$e_{1,3}$};
\node () at (-2,-3) {$\dots$};
\node[draw] (e1n) at (-1,-3) {$e_{1,n}$};
\node[draw] (e23) at (-0,-3) {$e_{2,3}$};
\node () at (1,-3) {$\dots$};
\node[draw] (en-1n) at (2,-3) {$e_{n-1,n}$};
\node[draw,rectangle,rounded corners, fit = (en-1n) (e12)]  () {};

\draw[color=black,fill=white,rounded corners] (-3.5,-5.4) rectangle (1.5,-4.6);
\node[draw] (v11) at (-2.8,-5) {$x_1$};
\node[draw] (v12) at (-1.8,-5) {$x_2$};
\node () at (-0.5,-5) {$\dots$};
\node[draw] (v1n) at (1,-5) {$x_n$};

\draw 	(t1) -- (r) -- (t2);
\draw (r) -- (tt);

\draw[color=black,fill=white,rounded corners=0.0cm] (-3.7,-2.3) rectangle (1.7,-1.7);
\node () at (-1,-2) {$a_ie_{uv} \in E(G_2) \Leftrightarrow uv \in E(G^c_i)$};

\draw (-1,-1.35) -- (-1,-1.7);
\draw (-1,-2.3) -- (-1,-2.63) ;

\draw[color=black,fill=white,rounded corners=0.0cm] (-4.3,-3.85) rectangle (2.3,-4.25);
\node () at (-1,-4.05) {$e_{uv}v_u, e_{uv}v_v \in E(G_2), \forall 1 \leqslant u < v \leqslant n$};

\draw (-1,-3.36) -- (-1,-3.85);
\draw (-1,-4.25) -- (-1,-4.6);


\end{tikzpicture}
\caption{Illustration of the construction of $G_2$.}\label{fig:g2}
\end{figure}
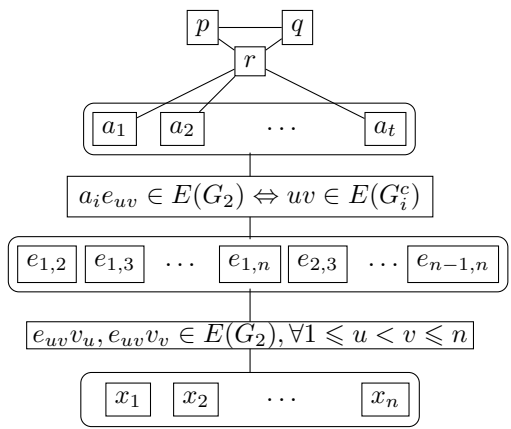

To build $G_1$ (see also Figure~\ref{fig:g1}):
\begin{itemize}
\item $V(G_1) = \{p,q,r,a\} \cup \{e_i $ $|$ $ 1 \leqslant i \leqslant \binom{l}{2}\} \cup \{x_i $ $|$ $ 1\leqslant i \leqslant l\}$,
\item $E(G_1)_1 = \{pq,pr,qr,ra\}$,
\item $E(G_1)_2 = \{ae_{i} $ $|$ $ 1 \leqslant i \leqslant \binom{l}{2} \}$,
\item $E(G_1)_3 = \{e_ix_u,e_ix_v $ $|$ $ \forall 1 \leqslant i \leqslant \binom{l}{2}, e_i = uv \}$,
\item $E(G_1) = E(G_1)_1 \cup E(G_1)_2 \cup E(G_1)_3$.
\end{itemize}

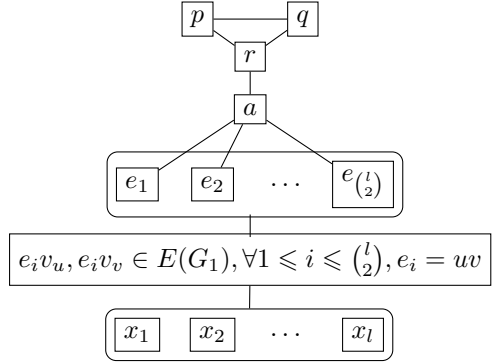
\begin{figure}[ht!]
\centering
\begin{tikzpicture}


\node[draw] (r2) at (-0.5,-1.3) {$r$};
\node[draw] (p2) at (-1.2,-0.8) {$p$};
\node[draw] (q2) at (0.2,-0.8) {$q$};

\draw (p2) -- (r2) -- (q2);
\draw (p2) -- (q2);

\node[draw] (ti) at (-0.5,-2) {$a$};

\node[draw] (e1) at (-2,-3) {$e_{1}$};
\node[draw] (e2) at (-1,-3) {$e_{2}$};
\node () at (0,-3) {$\dots$};
\node[draw] (e2l) at (1,-3) {$e_{\binom{l}{2}}$};
\node[fit=(e1)(e2)(e2l),draw,rounded corners] (e1g1) {};

\node[draw] (v1) at (-2,-5) {$x_1$};
\node[draw] (v2) at (-1,-5) {$x_2$};
\node () at (0,-5) {$\dots$};
\node[draw] (vl) at (1,-5) {$x_l$};
\node[fit=(v1)(v2)(vl),draw,rounded corners] (vg1) {};

\node[draw] (eg1) at (-0.5,-4) {$e_iv_u, e_iv_v \in E(G_1), \forall 1 \leqslant i \leqslant \binom{l}{2}, e_i=uv$};

\draw (-0.5,-3.4) -- (-0.5,-3.7);
\draw (vg1) -- (eg1);

\draw (r2) -- (ti) -- (e1);
\draw (e2) -- (ti) -- (e2l);
%



\end{tikzpicture}
\caption{Illustration of the construction of $G_1$.}\label{fig:g1}

\end{figure}

We set $l' = |V(G_1)|$, and $Z = \{p,r\} \cup \{e_{uv} | 1 \leqslant u < v \leqslant n\}$. It is easy to see that $Z$ is indeed a vertex cover for $G_2$ and that its size is equal to $\frac{n(n-1)}{2} + 2$, which is polynomial in $n$ and hence in the size of the largest instance. Note that the size of the graph $G_1$ does not depend on $t$ and is polynomial in $n$, so the size of its vertex cover is also polynomial in $n$ and independent of $t$.

Let us show that $G_1$ is an induced subgraph of $G_2$ iff at least one of the $G^c_i$'s has a clique of size $l$.

$(\Leftarrow)$ Suppose that $G^c_i$ has a clique of size $l$. We denote by $S \subseteq V(G^c_i)$ a clique of size exactly $l$ in $G^c_i$. We show that there is an induced subgraph $S'$ of $G_2$ of size $l'$, isomorphic to $G_1$. We set $V(S') = \{p,q,r\} \cup \{a_i\} \cup \{e_{uv} $ $|$ $\forall uv \in E(S)\} \cup \{x_u | u \in S\}$. One can easily check that this subgraph is isomorphic to $G_1$.

$(\Rightarrow)$ Assume now that $G_1$ is an induced subgraph of $G_2$. 
Denote by $S'$ the subgraph of $G_2$ isomorphic to $G_1$. 
Note that the only triangle in $G_2$ is $pqr$.
Indeed, $T(V(G_2) \setminus \{p\})$ is bipartite.
The triangle $pqr$ in $G_1$ has therefore to match $pqr$ in $G_2$.
Moreover, $r$ in $G_1$ has to match $r$ in $G_2$ since $p$ and $q$ have no edges besides the clique $pqr$.
The vertex $a$ in $G_1$ can only match a vertex $a_i$ for some $i \in \{1,\dots,t\}$.
Then, $e_1$ up to $e_{\binom{l}{2}}$ in $G_1$ has to match $\binom{l}{2}$ vertices in $\{e_{uv} $ $|$ $ 1 \leqslant u < v \leqslant n\}$ of $G_2$ which correspond to actual edges in $G^c_i$. 
Finally, $x_1$ up to $x_l$ in $G_1$ has to match $l$ vertices among the $x_j$'s in $G_2$.
Note that the number of edges in $E(G_1)_3$ between the $e_j$'s and the $x_j$'s is 
exactly $2\binom{l}{2}=l(l-1)$. 
More precisely, each $e_j$ touches $2$ edges in $E(G_1)_3$ and each $x_j$ touches $l-1$ edges in $E(G_1)_3$.
In order to get a match in $G_2$, one should find a set of $\binom{l}{2}$ edges inducing exactly $l$ vertices.
So, this set of $l$ vertices is a clique in $G^c_i$.
\end{proof}

Note that the parameter of MCIS in the previous reduction is exactly the size of $G_1$ and the graphs used in the proof are connected. 
Therefore, we have the following: 

\begin{corollary}
\isi and \mccis, parameterized by a bound on the minimum vertex covers of input graphs, do not have a polynomial-size kernel unless $\np \subseteq \mathsf{coNP}/poly$.
\end{corollary}

The algorithm of~\cite{AbuKhzam2014} is not single-exponential for parameter sum of the vertex cover numbers. 
In fact, we show that a single-exponential algorithm is very unlikely.
This is, to the best of our knowledge, the first result of this type for parameter vertex cover.

\begin{theorem}
Under the ETH, IS(C)I cannot be solved in time $2^{o(k \log k)}$ when parameter $k$ is the sum of the vertex cover number of both graphs.
\end{theorem}

\begin{proof}
We give a reduction from $k \times k$ \textsc{Permutation Clique} which linearly preserves the parameter $k$.
It is known that this problem does not admit an algorithm with running time $2^{o(k \log k)}$ unless the ETH fails \cite{Lokshtanov2011}.
In the $k \times k$ \textsc{Permutation Clique} problem, one is given a graph over the set of vertices $[k] \times [k]$ and the goal is to find a clique of size $k$ such that in each \emph{row} and in each \emph{column} exactly one vertex is part of the clique, where a \emph{row} is the set of vertices $\{(i,1),(i,2),\ldots,(i,k)\}$ for some $i \in [k]$, and a column is the set of vertices $\{(1,j),(2,j),\ldots,(k,j)\}$ for some $j \in [k]$.

We first describe how the graph $G_2$ is built from any instance $G=([k] \times [k],E)$ of $k \times k$ \textsc{Permutation Clique}.
For each row (resp. column) index $i \in [k]$, we add two vertices $r_i^1$ and $r_i^2$ (resp. $c_i^1$ and $c_i^2$) that we link by an edge.
For $j \in [2]$, we set $R_j=\{r_1^j,r_2^j,\ldots,r_k^j\}$ (resp. $C_j=\{c_1^j,c_2^j,\ldots,c_k^j\}$) and $R=R_1 \cup R_2$ (resp. $C=C_1 \cup C_2$).
Then, to distinguish row indices from column indices, we add a clique $D_{r}$ of size $6$, and we link one designated vertex $r$ of $D_r$ to all the vertices in $R$.
We also add a clique $D$ of size $5$ with a special vertex $v$ in $D$ such that $v$ is linked to all the vertices in $R_1 \cup C_1$.

Finally, for each edge $e=(i,j)(i',j')$ of $G$ with $i \neq i' $ and $j \neq j'$ \footnote{We ignore the other edges since they are not relevant in finding a permutation clique.}, we add a vertex $v(e,1)$ that we link to the four vertices $r_i^1$, $c_j^1$, $r_{i'}^2$, and $c_{j'}^2$, and a vertex $v(e,2)$ that we link to the four vertices $r_i^2$, $c_j^2$, $r_{i'}^1$, and $c_{j'}^1$.
This ends the construction of $G_2$ (see Figure~\ref{fig:overall-vc}).
The pattern $G_1$ depends only on $k$ and is defined as the graph one obtains following the above construction when $G$ have all the edges of the form $(i,i)(i',i')$ and no other edges (in other words, $G$ has a $k$-clique on the diagonal and nothing else). 

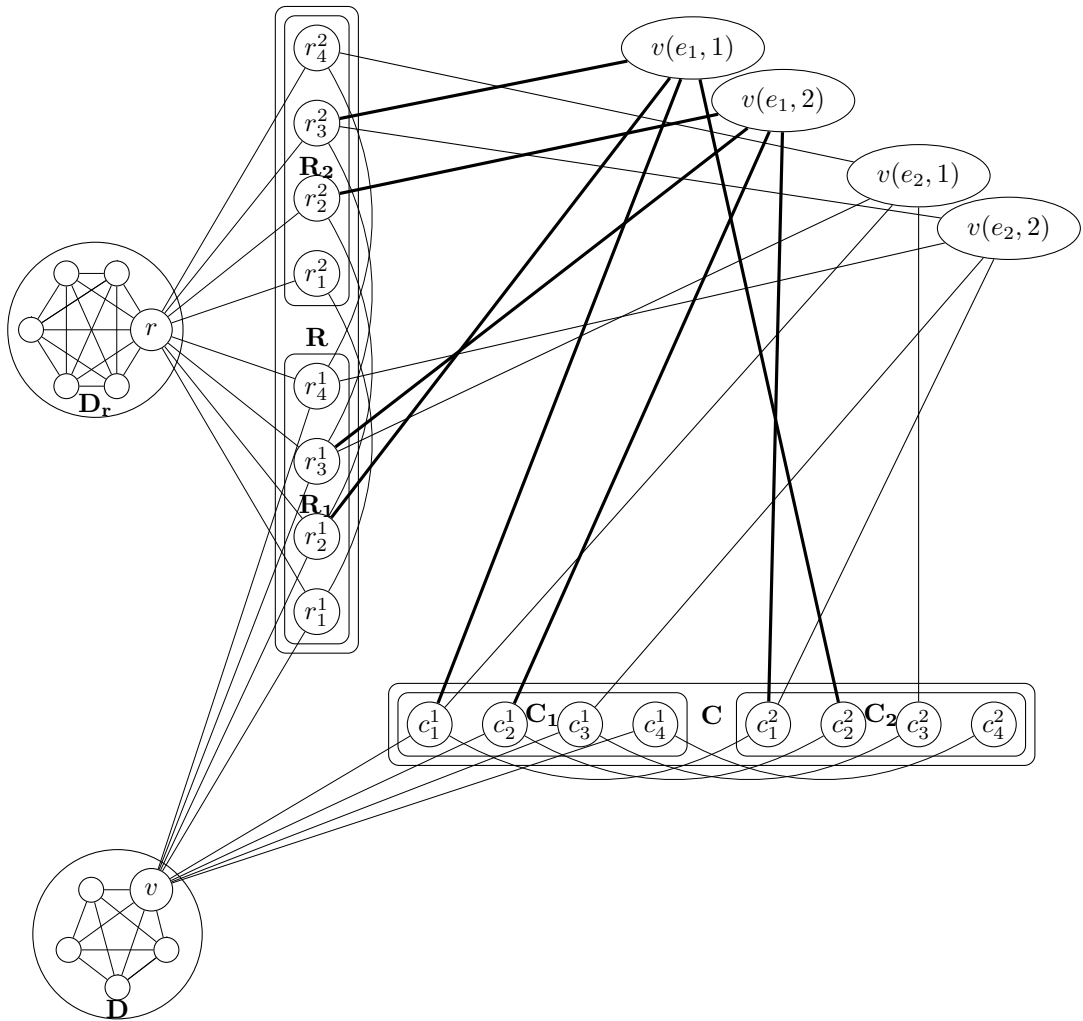
\begin{figure}
\centering
\begin{tikzpicture}[]
\def\k{4}
\node[draw,circle] (r) at (-2.2,\k+1.25) {$r$} ;
\node[draw,circle] (dr1) at (-3.8,\k+1.25) {} ;
\node[draw,circle] (dr2) at (-2.66,\k+2) {} ;
\node[draw,circle] (dr3) at (-3.33,\k+2) {} ;
\node[draw,circle] (dr4) at (-2.66,\k+0.5) {} ;
\node[draw,circle] (dr5) at (-3.33,\k+0.5) {} ;
\foreach \i [count=\s from 0] in {1,...,5}{
  \foreach \j in {1,2,...,\s}{
    \draw[very thin] (dr\i) -- (dr\j) ;
}
\draw[very thin] (r) -- (dr\i) ;
}
\node[draw,circle,fit=(r) (dr1) (dr2) (dr3) (dr4) (dr5),inner sep=-0.15cm] (Dr) {} ;
\node[below of=Dr] {$\mathbf{D_r}$} ;
\node[draw,circle] (v) at (-2.2,-2.2) {$v$} ;
\node[draw,circle] (d1) at (-2,-3) {} ;
\node[draw,circle] (d2) at (-2.65,-3.5) {} ;
\node[draw,circle] (d3) at (-3.3,-3) {} ;
\node[draw,circle] (d4) at (-3,-2.2) {} ;
\foreach \i [count=\s from 0] in {1,...,4}{
  \foreach \j in {1,2,...,\s}{
    \draw[very thin] (d\i) -- (d\j) ;
}
\draw[very thin] (v) -- (d\i) ;
}
\node[draw,circle,fit=(v) (d1) (d2) (d3) (d4),inner sep=-0.05cm] (D) {} ;
\node[below of=D] {$\mathbf{D}$} ;
\foreach \i in {1,...,\k}{
\foreach \j in {1,2}{
\node[draw,circle,inner sep=0.04cm] (r\i\j) at (0,\i+\k*\j-\k+0.5*\j) {$r^\j_\i$} ;
\draw[very thin] (r) -- (r\i\j) ;
}
\draw[very thin] (r\i1) to [bend right] (r\i2) ;
}
\foreach \j in {1,2}{
\node[draw,rectangle,rounded corners,fit=(r1\j) (r\k\j)] (R\j) {\\$\mathbf{R_\j}$} ;
}
\node[draw,rectangle,rounded corners,fit=(R1) (R2)] (R) {\\$\mathbf{R}$} ;
\foreach \i in {1,...,\k}{
\draw[very thin] (v) -- (r\i1) ;
}
\foreach \i in {1,...,\k}{
\foreach \j in {1,2}{
\node[draw,circle,inner sep=0.04cm] (c\i\j) at (\i+\k*\j-\k+0.5*\j,0) {$c^\j_\i$} ;
}
\draw[very thin] (c\i1) to [bend right] (c\i2) ;
}
\foreach \j in {1,2}{
\node[draw,rectangle,rounded corners,fit=(c1\j) (c\k\j)] (C\j) {$\mathbf{C_\j}$} ;
}
\node[draw,rectangle,rounded corners,fit=(C1) (C2)] (C) {$\mathbf{C}$} ;
\foreach \i in {1,...,\k}{
\draw[very thin] (v) -- (c\i1) ;
}
\node[draw,ellipse] (ve11) at (5,9) {$v(e_1,1)$} ;
\draw[very thick] (r21) -- (ve11) -- (c11) ;
\draw[very thick] (r32) -- (ve11) -- (c22) ;
\node[draw,ellipse] at (6.2,8.3) (ve12) {$v(e_1,2)$} ;
\draw[very thick] (r22) -- (ve12) -- (c12) ;
\draw[very thick] (r31) -- (ve12) -- (c21) ;

\begin{scope}[xshift=3cm,yshift=-1.7cm]
\node[draw,ellipse] (ve21) at (5,9) {$v(e_2,1)$} ;
\draw (r31) -- (ve21) -- (c11) ;
\draw (r42) -- (ve21) -- (c32) ;
\node[draw,ellipse] at (6.2,8.3) (ve22) {$v(e_2,2)$} ;
\draw (r32) -- (ve22) -- (c12) ;
\draw (r41) -- (ve22) -- (c31) ;
\end{scope}
\end{tikzpicture}
\caption{The overall construction of $G_2$. We represented only two edges of $G$: $e_1=(2,1)(3,2)$ and $e_2=(3,1)(4,3)$. For the sake of readability, the edges encoding $e_1$ are enhanced to distinguish them easily from the edges encoding $e_2$.}
\label{fig:overall-vc}
\end{figure}

Both $G_1$ and $G_2$ have $R \cup C \cup D_r \cup D$ as a vertex cover of size 
$4k+11$. $G_2$ has $|E|+4k+11=O(k^4)$ vertices and $G_1$ has 
$2{k \choose 2}+4k+11=O(k^2)$ vertices.
To avoid confusion about vertices in $G_1$ and $G_2$ we will denote the 
vertices and sets of vertices of $G_1$ with a tilde.
We now show that the reduction is valid.

Suppose there is a solution $\{(a_1,b_1), \ldots, (a_k,b_k)\}$ to the instance of $k \times k$ \textsc{Permutation Clique}.
Then, $G_1$ is an induced subgraph of $G_2$ with the following mapping.
We map $\tilde{r}$ to $r$ and $\tilde{v}$ to $v$.
We map $\tilde{D_r} \setminus \{\tilde{r}\}$ to $D_r \setminus \{r\}$ and $\tilde{D} \setminus \{\tilde{v}\}$ to $D \setminus \{v\}$ in an arbitrary way.
Then, for each $i \in [k]$ and $j \in [2]$, we map $\tilde{r}_i^j$ to $r_{a_i}^j$ and $\tilde{c}_i^j$ to $r_{b_i}^j$.
We observe that this mapping is one-to-one since $(a_1,b_1), \ldots, (a_k,b_k)$ is a \emph{permutation} clique, i.e., $\{a_1,a_2,\ldots,a_k\}=[k]=\{b_1,b_2,\ldots,b_k\}$.
Finally, for any $j \in [2]$, and any $i \neq i' \in [k]$ we map $\tilde{v}(e,j)$ to $v((a_i,b_i)(a_{i'},b_{i'}),j)$.
Note that vertex $v((a_i,b_i)(a_{i'},b_{i'}),j)$ always exists precisely because $\{(a_1,b_1), \ldots, (a_k,b_k)\}$ is a clique.

Conversely, if there is a solution to the IS(C)I instance, we will show that there is a permutation $k$-clique in $G$.
There is only one clique of size $6$ in $G_2$, so the clique $\tilde{D_r}$ of size $6$ has to be mapped to $D_r$.
Then, $\tilde{r}$, as the unique vertex of $\tilde{D_r}$ of degree larger than $5$, should be mapped to $r$.
Now, for the same reasons, $\tilde{D}$ should be mapped to $D$ and $\tilde{v}$ to $v$.
Vertices of $\tilde{R_1} \cup \tilde{C_1}$ are the only $2k$ unmatched vertices having $\tilde{v}$ as a neighbor, so those vertices should be matched to the only $2k$ unmatched vertices having $v$ as a neighbor, namely $R_1 \cup C_1$.
For similar reasons, $\tilde{R}$ should be mapped to $R$.
Now, $\tilde{R_2} \cup \tilde{C_2}$ can only be mapped to $R_2 \cup C_2$ as the only unmatched vertices having exactly one neighbor in $\tilde{R_1} \cup \tilde{C_1}$ ($R_1 \cup C_1$).

Thus, the $4k$ vertices of $\tilde{R} \cup \tilde{C}$ can only be mapped to $R \cup C$, such that for $j \in [2]$, $\tilde{R_j}$ is mapped to $R_j$ and $\tilde{C_j}$ is mapped to $C_j$.
The edges $\tilde{r}_i^1\tilde{r}_i^2$ and $r_i^1r_i^2$ (resp. $\tilde{c}_i^1\tilde{c}_i^2$ and $c_1^1c_i^2$) further constrains the mapping: if $\tilde{r}_i^1$ is mapped to $r_{i'}^1$ then $\tilde{r}_i^2$ has to be mapped to $r_{i'}^2$ (resp. if $\tilde{c}_i^1$ is mapped to $c_{i'}^1$ then $\tilde{c}_i^2$ has to be mapped to $c_{i'}^2$).
Hence, we can see the mapping from $\tilde{R} \cup \tilde{C}$ to $R \cup C$ as two permutations $\sigma_r$ and $\sigma_c$ on $k$ elements, such that for $j \in [2]$, for $i \in [k]$, $\tilde{r}_i^j$ is mapped to $r_{\sigma_r(i)}^j$ and $\tilde{c}_i^j$ is mapped to $c_{\sigma_c(i)}^j$.
Then, the current partial mapping can be extended to a solution only if $\{(\sigma_r(1),\sigma_c(1)),\ldots, (\sigma_r(k),\sigma_c(k))\}$ is a clique in $G$.
Indeed, $\forall j \in [2]$,  $\forall i \neq i' \in [k]$, $\tilde{v}((i,i)(i',i'),j)$ can only be mapped to a potential $v((\sigma_r(i),\sigma_c(i))(\sigma_r(i'),\sigma_c(i')),j)$ so that vertex has to exist, meaning that there should be an edge in $G$ between $(\sigma_r(i),\sigma_c(i))$ and $(\sigma_r(i'),\sigma_c(i'))$.

An algorithm solving IS(C)I in time $poly(|G_1|,|G_2|)2^{o(k \log k)}$ with $k := \text{vc}(G_1)+\text{vc}(G_2)$ would therefore translate into an algorithm running in time $2^{o(k \log k)}$ for $k \times k$ \textsc{Permutation Clique} and contradict the ETH.

\end{proof}

Despite the fact that ISI and MCIS have the same parameterized complexity with respect to the natural parameter, they exhibit different behaviors when considering structural parameters.
In fact, the latter is paraNP-hard when parameterized by the vertex cover of only one of the two graphs, whereas ISI is $\fpt$ when parameterized by the vertex cover of the second (host) graph.
To see this, note that when the host graph has a vertex cover of size $k$, the minimum size of a vertex cover in the pattern graph must be bounded by the parameter $k$; otherwise we have a NO-instance.
The claim follows from the fixed-parameter tractability of MCIS in this case~\cite{AbuKhzam2014}.  

%
%



\medskip




Given the negative result of Theorem~\ref{th:fvsone}, the next question we pose is whether MCCIS is in $\fpt$ with respect to the size of a minimum vertex cover. 
In \cite{AbuKhzam2014}, a parameterized algorithm is presented for MCIS when the parameter is a bound on the minimum vertex cover number of the input graphs.
However, that algorithm cannot help us much for solving MCCIS since it relies on the existence of a feasible solution of size at least $\approx n-k$ which consists of mapping the two \emph{big} independent sets of the two graphs onto each other.
Of course, this is not a feasible solution for MCCIS.
In the following we prove that MCCIS is fixed-parameter tractable w.r.t. $k := \text{vc}(G_1)+\text{vc}(G_2)$.


\begin{theorem}
\label{mccis_vc}
\mccis parameterized by $k := \text{vc}(G_1)+\text{vc}(G_2)$ is fixed-parameter tractable.
\end{theorem}

\begin{proof}
In time $O^*(2^k)$ (even $O^*(1.2738^k)$ \cite{CKX10}), we can find minimum vertex covers $C_1$ and $C_2$ in $G_1$ and $G_2$ respectively.
Let $I^{(j)}$ be the independent set $V(G_j) \setminus C_j$ for $j \in \{1,2\}$.
By assumption, our parameter $k$ is $\max(C_1,C_2)$, so we can enumerate all tripartitions of $C_1$ and $C_2$ in time $O^*(9^k)$.
We denote by $C_{1,m}$, $C_{1,u}$ and $C_{1,i}$ (respectively $C_{2,m}$, $C_{2,u}$ and $C_{2,i}$) the three sets of a tripartition of $C_1$ (respectively $C_2$).
For $j \in \{1,2\}$, $C_{j,u}$ corresponds to the vertices of $C_j$ that are not matched, so they may be deleted.
$C_{j,m}$ comprises the vertices matched to $C_{3-j,m}$ (that is, to the vertex cover of the other graph), and $C_{j,i}$ are the vertices matched to $I^{(3-j)}$, the independent set of the other graph. See Figure~\ref{vcvc}.

We observe that for $j \in \{1,2\}$, $I^{(j)}$ can be partitioned into at most $2^k$ classes of twins: $I^{(j)}_{1}, I^{(j)}_2, \ldots I^{(j)}_{2^k}$. A class of twins in this context is a set of vertices with an identical neighborhood in the vertex cover and there are at most $2^k$ subsets of $C_j$. Potentially, some classes can be empty: they correspond to a subset of the vertex cover $C_j$ that is not the (exact) neighborhood of any vertex in $I^{(j)}$.

At this point, we can enumerate the mappings between $C_{1,m}$ and $C_{2,m}$ in time $O^*(k^k)$ and the mappings between $C_{j,i}$ and $I^{(3-j)}$ in time $O^*((2^{k})^k)=O^*(2^{k^2})$.
Indeed, to match a vertex $u$ with a vertex $v$ or a twin of $v$ is equivalent.
Thus, in time $O^*((9k)^k 2^{k^2})$ we can enumerate all the solutions of MCIS where only vertices of $I^{(1)}$ could still be matched to vertices of $I^{(2)}$.
The optimal map of the independent sets can be done in polynomial time by matching the greatest number of vertices in each \emph{equivalent} twin class (which is the size of the smaller of the two equivalent twin classes), where a twin class $I^{(j)}_r$ in $I^{(j)}$ is equivalent to a twin class $I^{(3-j)}_s$ in $I^{(3-j)}$ if the vertices of $N(I^{(j)}_r) \setminus C_{j,u}$ and $N(I^{(3-j)}_s) \setminus C_{3-j,u}$ are in one-to-one correspondence.
\end{proof}

\begin{figure}[htb!]
\begin{center}
\begin{tikzpicture}

\draw[color=black,fill=white] (0,0) rectangle (3,4);
\draw (0,1.5) -- (3,1.5);
\draw[decorate,decoration=brace] (-0.2,0) -- (-0.2,1.45) node[midway,anchor=east,inner sep=3pt, outer sep=3pt]  {$C_1$};
\draw[decorate,decoration=brace] (-0.2,1.55) -- (-0.2,4) node[midway,anchor=east,inner sep=3pt, outer sep=3pt]  {$I^{(1)}=G_1[V_1 \setminus C_1]$};
\draw (0,0.5) -- (3,0.5);
\draw (0,1) -- (3,1);
\node() at (1.5, 0.25) {$C_{1,u}$};
\node() at (1.5, 0.75) {$C_{1,m}$};
\node() at (1.5, 1.25) {$C_{1,i}$};

\node () at (1.5,-0.5) {$G_1$};

\node[draw,circle] (i11) at (0.75,3.5) {};
\node[draw,circle] (i12) at (0.45,2.5) {};
\node[draw,circle] (i13) at (0.753,2) {};
\node[draw,rectangle,rounded corners,dashed,fit=(i11) (i12) (i13)] (i1) {$I^{(1)}_1$};

\node (d) at (1.5,3) {$\ldots$};

\node[draw,circle] (i2k1) at (2.3,3.5) {};
\node[draw,circle] (i2k2) at (2.5,2.2) {};
\node[draw,rectangle,rounded corners,dashed,fit=(i2k1) (i2k2)] (i2k) {$I_{2^k}^{(1)}$};

\node[draw,circle] (v1) at (0.3,1.25) {};
\node[draw,circle] (v2) at (2.2,1.25) {};
\node[draw,circle] (v2b) at (2.6,1.25) {};

\draw (v1) edge[bend left=30] (i11) ;
\draw (v1) -- (i12) ;
\draw (v1) -- (i13) ;

\draw (i2k1) edge[bend right=25] (v2) ;
\draw (i2k2) -- (v2) ;
\draw (i2k1) edge[bend left=30] (v2b) ;
\draw (i2k2) -- (v2b) ;

\begin{scope}[xshift=1cm]
\draw[color=black,fill=white] (3,0) rectangle (6,4);
\draw[decorate,decoration=brace] (6.2,1.45) -- (6.2,0) node[midway,anchor=west,inner sep=3pt, outer sep=3pt]  {$C_2$};
\draw[decorate,decoration=brace] (6.2,4) -- (6.2,1.55) node[midway,anchor=west,inner sep=3pt, outer sep=3pt]  {$I^{(2)}=G_2[V_2 \setminus C_2]$};
\draw (3,1.5) -- (6,1.5);
\draw (3,0.5) -- (6,0.5);
\draw (3,1) -- (6,1);
\node() at (4.5, 0.25) {$C_{2,u}$};
\node() at (4.5, 0.75) {$C_{2,m}$};
\node() at (4.5, 1.25) {$C_{2,i}$};
\node () at (4.5,-0.5) {$G_2$};

\node[draw,circle] (i11) at (3.5,3.5) {};
\node[draw,circle] (i13) at (3.753,2) {};
\node[draw,rectangle,rounded corners,dashed,fit=(i11)  (i13)] (i1) {$I^{(2)}_1$};

\node (d) at (4.5,3) {$\ldots$};

\node[draw,circle] (i2k1) at (5.4,3.5) {};
\node[draw,circle] (i2k2) at (5.6,2.2) {};
\node[draw,circle] (i2k3) at (5.2,2.2) {};
\node[draw,rectangle,rounded corners,dashed,fit=(i2k1) (i2k2) (i2k3)] (i2k) {$I_{2^k}^{(2)}$};

\node[draw,circle] (v1) at (3.4,1.25) {};
\node[draw,circle] (v1b) at (3.8,1.25) {};
\node[draw,circle] (v2) at (5.2,1.25) {};

\draw (v1) edge[bend left=15] (i11) ;
\draw (v1) -- (i13) ;
\draw (v1b) edge[bend right=30] (i11) ;
\draw (v1b) -- (i13) ;

\draw (i2k1) edge[bend right=30] (v2) ;
\draw (i2k2) -- (v2) ;
\draw (i2k3) -- (v2) ;
\draw (i2k1) -- (v1b) ;
\draw (i2k2) -- (v1b) ;
\draw (i2k3) -- (v1b) ;

\end{scope}

\draw (2.8,0.75) edge[<->] (4.2,0.75);
\draw (2.8,1.25) edge[<->] (4.2,2.75);
\draw (2.8,2.75) edge[<->] (4.2,1.25);
\draw (2.8,3) edge[<->] (4.2,3);

\end{tikzpicture}

\end{center}


\caption{Illustration of the proof of Theorem \ref{mccis_vc}. Dashed boxes represent the classes inside the independent set. Arrows represent the matching between sets of vertices. Sets $C_1$ (resp. $C_2$) represents a vertex cover for $G_1$ (resp. $G_2$).}
\label{vcvc}
\end{figure}
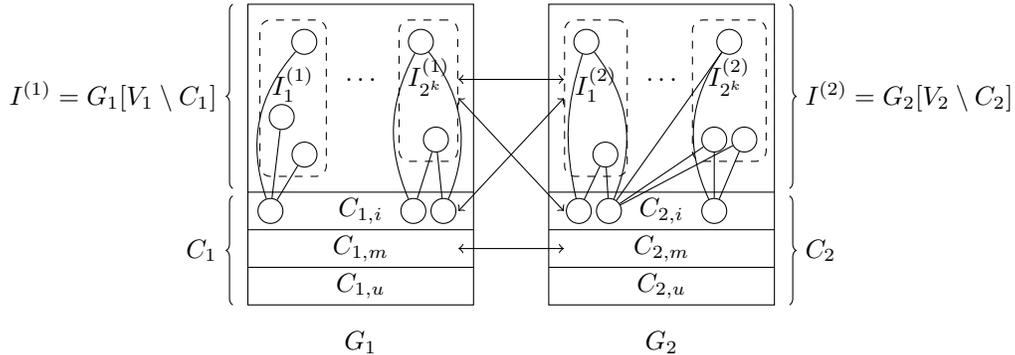


To find a solution for MCCIS, the algorithm described in the above proof enumerates all possible maximal common induced subgraphs in time $O^*((9k)^k 2^{k^2})$. 
The current bottleneck to improve it is when we try to match vertices of the vertex cover with vertices of the independent set.
For the not connected version of the problem, a trivial argument can bound the size of the independent set (if this one is big, there is a trivial solution), which cannot be used for the connected version.
As such, it can be used as an enumeration algorithm for MCIS.


\begin{corollary}
\mcis parameterized by $k := \text{vc}(G_1)+\text{vc}(G_2)$ is fixed-parameter enumerable.
\end{corollary}



Let us finish this section with some general considerations.
Note that for ISI, the parameter $\text{vc}+\text{fvs}$ is not the same 
as $\text{fvs}+\text{vc}$. 
In the latter, the parameter is a bound on the vertex cover of $G_2$ (as well as the feedback vertex set of $G_1$) which makes ISI in $\fpt$, while it remains open for $\text{vc}+\text{fvs}$. 
We also note that ISI is not in $\xp$ w.r.t. $\text{vc}(G_1)$ by a simple reduction from \textsc{Independent Set}: let $G_2$ be an edgeless graph on $k$ vertices, then its vertex cover number is 0.

We now turn our attention to the case where MCIS is parameterized by a 
combination of the natural parameter and some structural parameter. We note 
that, in general, such parameterization reduces the problem's complexity. This 
is most often due to the fixed-parameter tractability of MCIS in $H$-minor 
free graphs (again, since ISI is $\fpt$ in this case~\cite{FG01}).
For example, consider the case where the parameter is the sum of some bound $t$ 
on the feedback vertex set of the input graphs and the natural parameter $k$. 
The problem is $\fpt$ in this case since graphs of $t$-feedback vertex set are 
$H$-minor free where $H$ is the ``fixed'' graph consisting of a disjoint union 
of $t+1$ triangles. The same applies to parameterization by treewidth and the 
natural parameter by considering $H$ to be the complete graph on $t+2$ 
vertices, for example.


\section{Conclusion}

We studied the \mcis and \mccis problems with respect to the solution size on special graph classes such as forests, bipartite graphs, bounded degree graphs, bounded degeneracy graphs, graphs without small (length 3 or 4) cycles. 
The two problems are fixed-parameter tractable on $H$-minor free graphs, which include forests, and bounded degree graphs, but they are $\wone$-complete on bipartite graphs of girth 6 and degeneracy 2.
This ruled out at the same time two approaches to get fixed-parameter algorithms on subclasses of graphs for W-hard problems.

We then considered the use of structural parameters, such as a bound on the minimum vertex covers of the input graphs. 
Although both MCIS and MCCIS are in $\fpt$ in this case, we proved that no kernel of polynomial bound can be obtained unless $\np \subseteq \mathsf{coNP}/poly$ and that they cannot be solved in time $2^{o(k \log k)}$ under the ETH.
We noted that MCIS is not even in $\xp$ with respect to other (smaller) auxiliary parameters, such as treewidth and feedback vertex set. 
A few open problems remain to be addressed. 
For example, is MCIS/MCCIS in $\fpt$ when parameterized by the combination of the vertex cover number and the feedback vertex set number, or by the vertex cover number and the treewidth?
Moreover, it would be interesting to know whether the algorithm for MCCIS of Theorem~\ref{mccis_vc} can be improved to match the lower bound.




\bibliographystyle{abbrv}

\bibliography{mcis}

\end{document}